\newif\ifAMS
\AMStrue 
\ifAMS
\documentclass{amsart}
\else
\documentclass[proceedings,submission]{dmtcs}
\fi
\ifAMS 
\usepackage{graphicx}
\usepackage{url}
\fi
\usepackage[all]{xy}
\usepackage{qsymbols}
\usepackage{amsmath}
\usepackage{enumerate}
\usepackage{prooftree}

\newtheorem{fact}{Fact}
\newtheorem{prop}{Proposition}

\newtheorem{theo}{Theorem}
\newtheorem{corollary}{Corollary}
\newtheorem{conj}{Conjecture}
\newtheorem{definition}{Definition}

\newcommand{\ie}{i.e.,~}
\newcommand{\nat}{\ensuremath{\mathbb{N}}}
\newcommand{\Ss}{\ensuremath{\mathbb{S}}}
\newcommand{\St}{\ensuremath{\widetilde{\mathbb{S}}}}
\newcommand{\Tt}{\ensuremath{\mathbb{T}}}

\newcommand{\Vv}{\ensuremath{{\mathbb{S}_{\infty}}}}

\newcommand{\Var}[1]{\underline{\mathsf{#1}}}

\newcommand{\nbct}[1]{T_{0,#1}}
\newcommand{\nbt}[1]{T_{\infty,#1}}

\usepackage[usenames,dvipsnames]{color}
\definecolor{darkbrown}{cmyk}{.3,.75,.75,.15}
\definecolor{vertfonce}{rgb}{0,.5,0}

\definecolor{vertfonce}{rgb}{0,.5,0}

\newif\ifcomment

\commenttrue           % with comments
\ifAMS
\title[Counting Binary Lambda Terms]{Counting Terms in the Binary Lambda Calculus} 
\author[K. Grygiel,
P. Lescanne]{Katarzyna Grygiel$^\dagger$
  \\
  \and\\
  Pierre Lescanne$^{\dagger,\ddagger}$\\\\
  $^\dagger$Jagiellonian University,\\
  Faculty of Mathematics and Computer Science,\\
  Theoretical Computer Science Department, \\
  ul. Prof. {\L}ojasiewicza 6, 30-348 Krak\'ow, Poland\\\\
  $^\ddagger$University of Lyon, \\
  \'Ecole normale sup\'erieure de Lyon, \\
  LIP (UMR 5668 CNRS ENS Lyon UCBL INRIA)\\
  46 all\'ee d'Italie, 69364 Lyon, France\\
  \email{grygiel@tcs.uj.edu.pl,pierre.lescanne@ens-lyon.fr}}
\thanks{The first author was supported by the National Science Center of Poland,
  grant number 2011/01/B/HS1/00944, when the author hold a post-doc position at the
  Jagiellonian University within the SET project co-financed by the European Union.}
\else
\title[Counting Binary Lambda Terms]{Counting Terms in the Binary Lambda Calculus
  (extended abstract)} 
\author[K. Grygiel,P. Lescanne]{Katarzyna Grygiel$^\ddagger$\thanks{Supported by the National Science Center of Poland,
  grant number 2011/01/B/HS1/00944, when the author hold a post-doc position at the
  Jagiellonian University within the SET project co-financed by the European Union.},Pierre Lescanne$^{\ddagger,\star}$}
  \address{$^\ddagger$Jagiellonian University,\\
  Faculty of Mathematics and Computer Science,\\
  Theoretical Computer Science Department, \\
  ul. Prof. {\L}ojasiewicza 6, 30-348 Krak\'ow, Poland\\\\
  $^\star$University of Lyon, \\
  \'Ecole normale sup\'erieure de Lyon, \\
  LIP (UMR 5668 CNRS ENS Lyon UCBL INRIA)\\
  46 all\'ee d'Italie, 69364 Lyon, France\\
  \email{grygiel@tcs.uj.edu.pl,pierre.lescanne@ens-lyon.fr}
}
\fi
\begin{document}

\maketitle

\begin{abstract}
  In a paper entitled \emph{Binary lambda calculus and combinatory logic}, John Tromp
  presents a simple way of encoding lambda calculus terms as binary sequences. In
  what follows, we study the numbers of binary strings of a given size that represent
  lambda terms and derive results from their generating functions, especially that
  the number of terms of size $n$ grows roughly like $1.963447954^n$.
\end{abstract}

\ifAMS
\else
\fi

\section{Introduction}

In recent years growing attention has been given to quantitative research in logic
and computational models. Investigated objects (e.g., propositional formulae,
tautologies, proofs, programs) can be seen as combinatorial structures, providing
therefore the inspiration for combinatorists and computer scientists. In particular,
several works have been devoted to studying properties of lambda calculus terms. On a
practical point of view, generation of random lambda terms is the core of debugging
functional programs using random tests~\cite{DBLP:conf/icfp/ClaessenH00} and the
present paper offers an answer to a open question (see introduction
of~\cite{DBLP:conf/icfp/ClaessenH00}) since we are able to generate closed typable
terms following a uniform distribution.  This work applies beyond $`l$-calculus to
any system with bound variables, like first order predicate calculus (quantifiers are
binders like $`l$) or block structures in programming languages.

First traces of the combinatorial approach to lambda calculus date back to the work
of Jue Wang~\cite{wang04:_effic_gener_random_progr_their_applic}, who initiated the
idea of enumerating $`l$-terms. In her report, Wang defined the size of a term as the
total number of abstractions, applications and occurrences of variables, which
corresponds to the number of all vertices in the tree representing the given term.

This size model, although natural from the combinatorial viewpoint, turned out to be
difficult to handle. The question that arises immediately concerns the number of
$`l$-terms of a given size. This non-trivial task has been done by Bodini, Gardy, and
Gittenberger in \cite{bodini11:_lambd_bound_unary_heigh,2013arXiv1305.0640B} and
Lescanne in \cite{DBLP:journals/tcs/Lescanne13}.

The approach applied in the latter paper has been extended
in~\cite{grygiel_lescanne_jfp} by the authors of the current paper to the model in
which applications and abstractions are the only ones that contribute to the size of
a $`l$-term. The same model has been studied in
\cite{DBLP:journals/corr/abs-0903-5505} by David et al., where several properties
satisfied by random $`l$-terms are provided.

When dealing with the two described models, it is not difficult to define recurrence
relations for the number of $`l$-terms of a given size. However, by applying standard
tools of the theory of generating functions one obtains generating functions that are
expressed as infinite sequences of radicals. Moreover, the radii of convergence are
in both cases equal to zero, which makes the analysis of those functions very
difficult to cope with.

In this paper, we study the binary encoding of lambda calculus introduced by John
Tromp in \cite{DBLP:conf/dagstuhl/Tromp06}. This representation results in another
size model. It comes from the binary lambda calculus he defined in which he builds a
minimal self interpreter of lambda calculus\footnote{an alternative to universal
  Turing machine} as a basis of algorithmic complexity theory~\cite{LiVitanyi}.  Set
as a central question of theoretical computer science and mathematics, this approach
is also more realistic for functional programming.  Indeed for compiler builders it
is counter-intuitive to assign the same size to all the variables, because in the
translation of a program written in \textsf{Haskell}, \textsf{Ocaml} or \textsf{LISP}
variables are put in a stack.  A variable deep in the stack is not as easily
reachable as a variable shallow in the stack.  Therefore the weight of the former
should be larger than the weight of the latter. Hence it makes sense to associate a
size with a variable proportional to its distance to its binder.  In this model,
recurrence relations for the number of terms of a given size are built using this
specific notion of size.  From that, we derive corresponding generating functions
defined as infinitely nested radicals. However, this time the radius of convergence
is positive and allows us for further analysis of the functions. We are able to
compute the exact asymptotics for the number of all (not necessarily closed) terms
and we also prove the approximate asymptotics for the number of closed
ones. Moreover, we define an unranking function, \ie a generator of terms from their
indices from which we derive a uniform generator of $`l$-terms (general and typable).
This allows us to provide outcomes of computer experiments in which we estimate the
number of simply typable $`l$-terms of a given size.

\section{Lambda calculus and its binary representation}

Lambda calculus is a model of computation that is equivalent to Turing machines or
recursive functions, serving as a powerful tool in the development of the programming
theory~\cite{Mitchell96}. Furthermore, it constitutes the basis for functional
programming languages and has many applications in automated theorem provers.

Basic objects of the lambda calculus are $`l$-terms, which are regarded as denotation
for functions or computer programs. Given a countable infinite set of variables $V$,
we define lambda terms by the following grammar:
\[ M := V \; | \; `l V.M \; | \; (MM) .\]

A term of the form $`l x.M$ is called an abstraction. Each occurrence of $x$ in $M$
is called bound. We say that a variable $x$ is free in a term $N$ if it is not bound
by an enclosing abstraction. A term with no free variable is called closed. Two terms
are considered equivalent if they are identical up to renaming of bound variables.

In order to eliminate names of variables from the notation of a $`l$-term, de Bruijn
introduced an alternative way of representing equivalent terms. Instead of variables
we are given now a set of de Bruijn indices $\{ \Var{1}, \Var{2}, \Var{3}, \ldots
\}$. Given a closed $`l$-term, we form the corresponding de Bruijn term as follows:
an abstraction $`l x.M$ is now written as $`l \Var{M}$, where $\Var{M}$ is the result
of substituting each occurrence of $x$ by the index $\Var{n}$, where $n$ is the
number of $`l$'s enclosing the given occurrence of $x$; an application $MN$ is simply
replaced by $\Var{M}\Var{N}$.
% de Bruijn

Following John Tromp, we define the binary representation of de Bruijn indices in the
following way:
\begin{eqnarray*}
  \widehat{\lambda \Var{M}} &=& 00\widehat{\Var{M}},\\
  \widehat{\Var{M}\ \Var{N}} &=& 01\widehat{\Var{M}}\widehat{\Var{N}},\\
  \widehat{\Var{i}} &=& 1^i0.
\end{eqnarray*}
However, notice that unlike Tromp~\cite{DBLP:conf/dagstuhl/Tromp06} and
Lescanne~\cite{LescannePOPL94}, we start the de Bruijn indices at $1$ like de
Bruijn~\cite{deBruijn72}.  Given a $`l$-term, we define its size as the length of the
corresponding binary sequence, \ie
\begin{eqnarray*}
  |\Var{n}| &=& n +1,\\
  |`l M | &=& |M| + 2,\\
  |M\,N| &=& |M| + |N|+2.
\end{eqnarray*}

In contrast to previously studied models, the number of all (not necessarily closed)
$`l$-terms of a given size is always finite. This is due to the fact that the size of
each variable depends on the distance from its binder.

\section{Combinatorial facts}
In order to determine the asymptotics of the number of all/closed $`l$-terms of a
given size, we will use the following combinatorial notions and results.

We say that a sequence $(F_n)_{n \geq 0}$ is of
\begin{itemize}
\item order $G_n$, for some sequence $(G_n)_{n \geq 0}$ (with $G_n\neq 0$), if
  \[ \lim_{n \to \infty} F_n/G_n = 1, \] and we denote this fact by $F_n \sim G_n$;
\item exponential order $A^n$, for some constant $A$, if
  \[ \limsup_{n \to \infty} |F_n|^{1/n} = A, \] and we denote this fact by $F_n
  \bowtie A^n$.
\end{itemize}

Given the generating function $F(z)$ for the sequence $(F_n)_{n \geq 0}$, we write
$[z^n]F(z)$ to denote the $n$-th coefficient of the Taylor expansion of $F(z)$,
therefore $[z^n] F(z) = F_n$.

The theorems below (Theorem IV.7 and Theorem VI.1 of \cite{flajolet08:_analy_combin})
serve as powerful tools that allow to estimate coefficients of certain functions that
frequently appear in combinatorial considerations.

\begin{fact}
  If $F(z)$ is analytic at $0$ and $R$ is the modulus of a singularity nearest to the
  origin, then
  \[ [z^n]F(z) \bowtie (1/R)^n.\]
\end{fact}

\begin{fact}\label{fact:asym_exp}
  Let $\alpha$ be an arbitrary complex number in $\mathbb{C}\setminus
  \mathbb{Z}_{\leq 0}$. The coefficient of $z^n$ in
  \[ f(z) = (1-z)^{\alpha} \] admits the following asymptotic expansion:
  \begin{eqnarray*} [z^n]f(z) &\sim& \frac{n^{\alpha - 1}}{\Gamma (\alpha)} \left( 1
      + \frac{\alpha (\alpha - 1)}{2n} + \frac{`a(`a-1)(`a-2)(3`a-1)}{24n^2} + O
      \left( \frac{1}{n^{3}} \right) \right),
    % \right. \\
    % && \qquad\qquad \left. + \frac{`a^2(`a-1)^2(`a -2)(`a-3)}{48n^3}+ O \left(
    %     \frac{1}{n^{4}} \right) \right) ,
  \end{eqnarray*}
  where $\Gamma$ is the Euler Gamma function.% defined for $\Re(\alpha) > 0$ as
  % \[ \Gamma ( \alpha ) := \int_{0}^{\infty} e^{-t} t^{\alpha -1 } dt .\]
\end{fact}

\section{The sequences $S_{m,n}$}
  
Let us denote the number of $`l$-terms of size $n$ with at most $m$ distinct free
indices by $S_{m,n}$.

First, let us notice that there are no terms of size $0$ and $1$. Let us consider a
$`l$-term of size $n+2$ with at most $m$ distinct free variables. Then we have one of
the following cases.
\begin{itemize}
\item The term is a de Bruijn index $\Var{n+1}$, provided $m$ is greater than or
  equal to $n+1$.
\item The term is an abstraction whose binary representation is given by
  $00\widehat{M}$, where the size of $M$ is $n$ and $M$ has at most $m+1$ distinct
  free variables.
\item The term is an application whose binary representation is given by
  $01\widehat{M}\widehat{N}$, where $M$ is of size $i$ and $N$ is of size $n-i$, with
  $i \in \{ 0,\ldots,n \}$, and both terms have at most $m$ distinct free variables.
\end{itemize}

This leads to the following recursive formula\footnote{Given a predicate $P$,
  $[P(\vec{x})]$ denotes the Iverson symbol, i.e., $[P(\vec{x})] = 1$ if $P(\vec{x})$
  and $[P(\vec{x})] = 0$ if $\neg P(\vec{x})$.}:
\begin{eqnarray}
  S_{m,0} &=& S_{m,1} ~=~ 0,\label{eq:Smn}\\
  S_{m,n+2} &=& [m \ge n+1] + S_{m+1,n} + \sum_{k=0}^n S_{m,k} S_{m,n-k}.\label{eq:Smn2}
\end{eqnarray}
The sequence $S_{0,n}$, \ie the sequence of numbers of closed $`l$-terms of size $n$,
can be found in the \emph{On-line Encyclopedia of Integer Sequences} under the number
\textbf{A114852}. Its first $20$ values are as follows:
\[ 0,\ 0,\ 0,\ 0,\ 1,\ 0,\ 1,\ 1,\ 2,\ 1,\ 6,\ 5,\ 13,\ 14,\ 37,\ 44,\ 101,\ 134,\
298,\ 431.\]

Now let us define the family of generating functions for sequences $(S_{m,n})_{n \geq
  0}$:
\begin{eqnarray*}
  \Ss_m(z) &=& \sum_{n=0}^{\infty} S_{m,n}\, z^n.
\end{eqnarray*}

Most of all, we are interested in the generating function for the number of closed
terms, \ie
\begin{eqnarray*}
  \Ss_0(z) &=& \sum_{n=0}^{\infty} S_{0,n}\, z^n.
\end{eqnarray*}

Applying the recurrence on $S_{m,n}$, we get
\begin{eqnarray*}
  \Ss_m(z) &=& z^2 \sum_{n= 0}^\infty S_{m,n+2} z^n\\
  &=& z^2 \sum_{n= 0}^\infty [m \ge n+1]z^n + z^2 \sum_{n= 0}^\infty S_{m+1,n}\,z^n + z^2 \sum_{n= 0}^\infty\sum_{k=0}^n
  S_{m,k} S_{m,n-k}\,z^n\\
  &=&  z^2 \sum_{k=0}^{m-1} z^k + z^2 \Ss_{m+1}(z) + z^2 \Ss_m(z)^2\\
  &=& \frac{z^2\,(1-z^m)}{1-z}  + z^2 \Ss_{m+1}(z) + z^2 \Ss_m(z)^2.
\end{eqnarray*}

Solving the equation
\begin{eqnarray}\label{eq:Szu}
  z^2 \Ss_m(z)^2 -  \Ss_m(z) +\frac{z^2\,(1-z^m)}{1-z}  + z^2 \Ss_{m+1}(z) = 0
\end{eqnarray}
gives us
\begin{eqnarray}\label{eq:1}
  \Ss_m(z) \ =\ \frac{1 - \sqrt{1 - 4z^4 \left(\frac{1-z^m}{1-z}  +  \Ss_{m+1}(z)\right)}}{2 z^2}.
\end{eqnarray}

This means that the generating function $\Ss_m(z)$ is expressed by means of
infinitely many nested radicals, a phenomenon which has already been encountered in
previous research papers on enumeration of lambda terms, see e.g.,
\cite{bodini11:_lambd_bound_unary_heigh}. However, in Tromp's binary lambda calculus
we are able to provide more results than in other representations of lambda terms.

First of all, let us notice that the number of lambda terms of size $n$ has to be
less than $2^n$, the number of all binary sequences of size $n$. This means that in
the considered model of lambda terms the radius of convergence of the generating
function enumerating closed lambda terms is positive (even larger that $1/2$), which
is not the case in other models, where the radius of convergence is equal to zero.

\section{The number of all $`l$-terms}
% {The function $\Vv(z)$}

Let us now consider the sequence enumerating all binary $`l$-terms, \ie including
terms that are not closed. Let $S_{\infty,n}$ denote the number of all such terms of
size $n$. Repeating the reasoning from the previous section, we obtain the following
recurrence relation:

\begin{eqnarray*}
  S_{\infty,0} &=& S_{\infty,1} ~=~ 0,\\
  S_{\infty,n+2} &=& 1 + S_{\infty,n} + \sum_{k=0}^n S_{\infty,k} S_{\infty,n-k}.
\end{eqnarray*}

The sequence $(S_{\infty,n})_{n`:\nat}$ can be found in \emph{On-line Encyclopedia of
  Integer Sequences} with the entry number \textbf{A114851}. Its first $20$ values
are as follows:
\[ 0,\ 0,\ 1,\ 1,\ 2,\ 2,\ 4,\ 5,\ 10,\ 14,\ 27,\ 41,\ 78,\ 126,\ 237,\ 399,\ 745,\
1292,\ 2404,\ 4259.\]
 
Obviously, we have $S_{m,n} \le S_{\infty,n}$ for every $m,n \in \nat$. Moreover,
$\displaystyle \lim_{m\to \infty} S_{m,n} = S_{\infty,n}$.

Let $\Vv(z)$ denote the generating function for the sequence
$(S_{\infty,n})_{n`:\nat}$, that is
\[\Vv(z) = \sum_{n=0}^{\infty}S_{\infty,n} z^n.\]
Notice that for $m \geq n-1$ we have $S_{m,n}=S_{\infty,n}$. Therefore
\[\Vv(z) = \sum_{n=1}^{\infty}S_{n,n}z^{n},\]
which yields that $[z^n]\Ss_{n,n}=[z^n]\Ss_{\infty,n}$. Furthermore, $\displaystyle
\Vv(z) = \lim_{m \to \infty} \Ss_m(z)$.

\begin{theo}
  The number of all binary $`l$-terms of size $n$ satisfies
  \[ S_{\infty,n} \sim (1/\rho)^n \cdot \frac{C}{n^{3/2}},\] where $\rho \doteq
  0.509308127$ and $C \doteq 1.021874073$.
\end{theo}

\ifAMS
\begin{proof}
  The generating function $\Vv(z)$ fulfills the equation
  \[\Vv(z) = \frac{z^2}{1-z} + z^2 \Vv(z) + z^2 \Vv(z).\]

  Solving the above equation gives us
  \[\Vv(z) = \frac{ z^3 - z^2 - z + 1 - \sqrt{z^6 + 2\,z^5 - 5\,z^4 + 4\,z^3 - z^2 -
      2\,z + 1}}{2z^2(1 - z)}.\]

  The dominant singularity of the function $\Vv(z)$ is given by the root of smallest
  modulus of the polynomial
  \[R_\infty(z) = z^6 + 2\,z^5 - 5\,z^4 + 4\,z^3 - z^2 - 2\,z + 1.\]

  The polynomial has four real roots: \[0.509308127, \quad -0.623845142, \quad 1,
  \quad -3.668100004,\] and two complex ones that are approximately equal to $0.4 +
  0.8i$ and $0.4 - 0.8i$.
  % as shown in Figure~\ref{fig:roots}.

  % \begin{figure}[htb!]
  %   \centering
  %   \includegraphics[scale=.5]{roots.png}
  %   \caption{Position of the roots of polynomial $z^5 + 3z^4 - 2z^3 + 2z^2 + z -
  %   1$}
  %   \label{fig:roots}
  % \end{figure}

  Therefore $\rho \doteq 0.509308127$ is the singularity of $\Vv$ nearest to the
  origin.  Let us write $\Vv(z)$ in the following form:
  \[\Vv(z) = \frac{ 1-z^2 - \sqrt{\rho (1-\frac{z}{\rho}) \cdot
      \frac{Q(z)}{1-z}}}{2z^2},\] where $Q(z)=\frac{R_\infty(z)}{(\rho-z)(1-z)}$ is
  the polynomial defined for all $|z| \le \rho$.

  We get that the radius of convergence of $\Vv(z)$ is equal to $\rho$ and its
  inverse $\frac{1}{`r} \doteq 1.963447954$ gives the growth of
  $S_{\infty,n}$. Hence, $S_{\infty,n} \bowtie (1/\rho)^n$.
 
  Fact \ref{fact:asym_exp} allows us to determine the subexponential factor of the
  asymptotic estimation of the number of terms. Applying it, we obtain that
  \[ [z^n]\Vv(z) \sim \left( \frac{1}{\rho}\right)^n \cdot \frac{n^{-3/2}}{\Gamma
    (-\frac{1}{2})} \cdot \widetilde{C}, \] where the constant $\widetilde{C}$ is
  given by
  \[\widetilde{C} = \frac{- \sqrt{\rho \cdot \frac{Q(\rho)}{1-\rho}}}{2 \rho ^{2}}
  \doteq -0.288265354.\] Since $\displaystyle \frac{\widetilde{C}}{\Gamma
    (-\frac{1}{2})} \doteq 1.021874073$, the theorem is proved.
\end{proof}
\else
The proof is in the full version of the paper~\cite{GrygielLescanne-Binary}. Notice
that $1/`r \doteq 1.963447954$.
\fi
\section{The number of closed $`l$-terms}
% {The dominant singularity of the function $\Ss_0(z)$}

\begin{prop}
  Let $\rho_m$ denote the dominant singularity of $\Ss_m(z)$. Then for every natural
  number $m$ we have
  \[\rho_m = \rho_0,\]
  which means that all functions $\Ss_m(z)$ have the same dominant singularity.
\end{prop}

\begin{proof}
  First, let us notice that for every $m,n \in \nat$ we have $S_{m,n} \le
  S_{m+1,n}$. This means that the radius of convergence of the generating function
  for the sequence $(S_{m,n})_{n \in \nat}$ is not smaller that the radius of
  convergence of the generating function for $(S_{m+1,n})_{n \in \nat}$. Therefore,
  for every natural number $m$, we have
  \[ \rho_m \ge \rho_{m+1} .\]

  On the other hand, from Equation~\ref{eq:1} we see that every singularity of
  $\Ss_{m+1}(z)$ is also a singularity of $\Ss_m(z)$. Hence, the dominant singularity
  of $\Ss_{m}(z)$ is less than or equal to the dominant singularity of
  $\Ss_{m+1}(z)$, \ie we have
  \[ \rho_m \le \rho_{m+1} .\]

  These two inequalities show that dominant singularities of all functions $\Ss_m(z)$
  are the same. In particular, for every $m$ we have $\rho_m = \rho_0$.
\end{proof}

\begin{prop}
  The dominant singularity of $\Ss_0(z)$ is equal to the dominant singularity of
  $\Vv(z)$, \ie
  \[\rho_0 \ = \ \rho \ \doteq \ 0.509308127. \]
\end{prop}

\begin{proof}
  Since the number of closed binary $`l$-terms is not greater than the number of all
  binary terms of the same size, we conclude immediately that $\rho_0 \geq \rho$.

  Let us now consider the functionals
  \begin{eqnarray*}
    `F_m(F) &=& \frac{1-\sqrt{1- 4z^4(\frac{1-z^m}{1-z} + F)}}{2z^2},\\
    `F_{\infty}(F) &=& \frac{1-\sqrt{1- 4z^4(\frac{1}{1-z} + F)}}{2z^2}.
  \end{eqnarray*}

  In particular, when $m=0$, we have
  \begin{eqnarray*}
    `F_0(F) &=&  \frac{1-\sqrt{1+ 4z^4F}}{2z^2}.
  \end{eqnarray*}
  We have also
  \begin{eqnarray*}
    \Ss_m(z) &=& `F_m(\Ss_{m+1}(z)).
  \end{eqnarray*}

  The $`F_m$'s and $`F_{\infty }$ are increasing over functions over $[0,1)$, which
  means that
  \begin{eqnarray*}
    F \le G &"=>"&`F_m(F) \le `F_m(G),\\
    F \le G &"=>"&`F_{\infty}(F) \le `F_{\infty}(G).
  \end{eqnarray*}

  For each $m \in \nat$, let us consider the function $\St_m(z)$ defined as the fixed
  point of $`F_m$. In other words, $\St_m(z)$ is defined as the solution of the
  following equation:
  \begin{eqnarray*}
    \St_m(z) &=& `F_m(\St_m(z)).
  \end{eqnarray*}

  Notice that since $S_{m,n} \le S_{m+1,n} \le S_{\infty,n}$ we can claim that
  $\Ss_m(z) \le\Ss_{m+1}(z) \le \Ss_\infty(z)$. Therefore, we have
  \begin{eqnarray}
    `F_m(\Ss_m(z)) &\le & \Ss_m(z), \\
    \St_m(z) & \le & \Ss_m(z) \ \le \ \Vv(z). \label{eq:tilde}
  \end{eqnarray}

  Since $\St_m(z)$ satisfies
  \begin{eqnarray*}
    2z^2\St_m(z) &=& 1 -  \sqrt{1 - 4z^4 \Big( \frac{1-z^m}{1-z} + \St_m(z) \Big)},%\\
    % (1 -  2z^2\St_m(z))^2 &=& 1 - 4z^4(\frac{1-z^m}{1-z} - \St_m(z))\\
    % - \St_m(z) + z^2\St_m^2(z) &=& -z^2(\frac{1-z^m}{1-z} - \St_m(z)),
  \end{eqnarray*}
  we get
  \begin{eqnarray*}
    z^2\St_m^2(z) - (1-z^2)^2 \St_m(z) + \frac{z^2(1-z^m)}{1-z} = 0.
  \end{eqnarray*}

  The discriminant of this equation is:
  \begin{eqnarray*}
    `D_m &=& (1-z^2)^2 - \frac{4z^4(1-z^m)}{1-z}.
  \end{eqnarray*}

  The values for which $`D_m = 0$ are the singularities of $\St_m(z)$. Let us denote
  the main singularity of $\St_m(z)$ by $`s_m$. From Equation~\eqref{eq:tilde} we see
  that
  \begin{eqnarray*}
    `s_m\ge `r_m \ge `r. \label{eq:sigma_m_rho_m_rho}
  \end{eqnarray*}

  The value of $`s_m$ is equal to the root of smallest modulus of the following
  polynomial:
  \[ P_m(z) \ := \ (z-1)`D_m \ = \ 4z^4(1-z^m) - (1-z)^3 (1+z)^2 .\]

  In the case of the function $\St_\infty(z)$, we get the polynomial
  \begin{eqnarray*}
    P_\infty(z)&=&z^5 + 3z^4 - 2z^3 + 2z^2 + z - 1 \ =\ \frac{R_\infty(z)}{z-1},
  \end{eqnarray*}
  whose root of smallest modulus is, like in the case of $R_\infty(z)$, equal to
  $`r$.

  Now let us show that the sequence $(\sigma_m)_{m \in \nat}$ of roots of polynomials
  $P_m(z)$ is decreasing (see Figure~\ref{fig:roots}) and that it converges to $`r$.
  \begin{figure}[t!h]
    \centering
    \includegraphics[scale=.5]{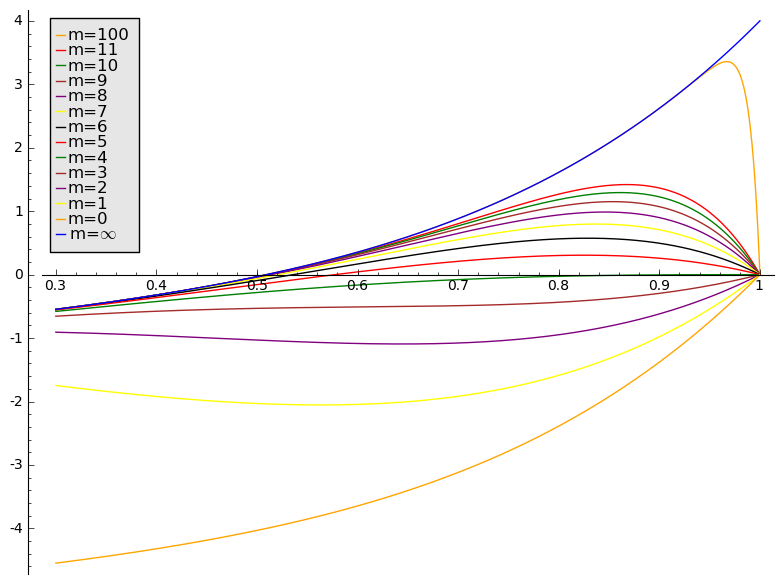}
    \caption{Roots of the $P_m$'s}
    \label{fig:roots}
  \end{figure}

  Notice that $P_m(z) = P_\infty(z) - 4z^{m + 4}$. Given a value $`z$ such that
  $`r<`z<1$ (for instance $`z=0.8$), $P_m(z)$ converges uniformly to $P_\infty(z)$ in
  the interval $[0,`z]$. Therefore $`s_m "->" `r$ when $m"->"\infty$. By $`s_m \leq
  `r_m \leq `r$, we get $`r_m "->" `r$, as well. Since all the $`r_m$'s are equal, we
  obtain that $`r_m =`r$ for every natural $m$.
\end{proof}

The above proposition leads immediately to the following result.

\begin{corollary}
  The number of closed binary $`l$-terms of size $n$ is of exponential order
  $(1/\rho)^n$, \ie
  \[ S_{0,n} \bowtie 1.963448^n.\]
\end{corollary}

The number of closed terms of a given size cannot be greater than the number of all
terms. Therefore, we obtain what follows.

\begin{theo}
  The number of closed binary $`l$-terms of size $n$ is asymptotically of order
  \[S_{0,n} \sim \left( \frac{1}{\rho} \right)^n \cdot O \big( n^{-3/2} \big) \doteq
  1.963448^n \cdot O \big( n^{-3/2} \big). \]
\end{theo}

Figure \ref{fig:Smn_ran_n32} shows values $S_{m,n} \cdot `r^n \cdot n^{3/2}$ for a
few initial values of $m$ and $n$ up to $600$.

\begin{figure}[t!h]
  \centering
  \includegraphics[scale=.5]{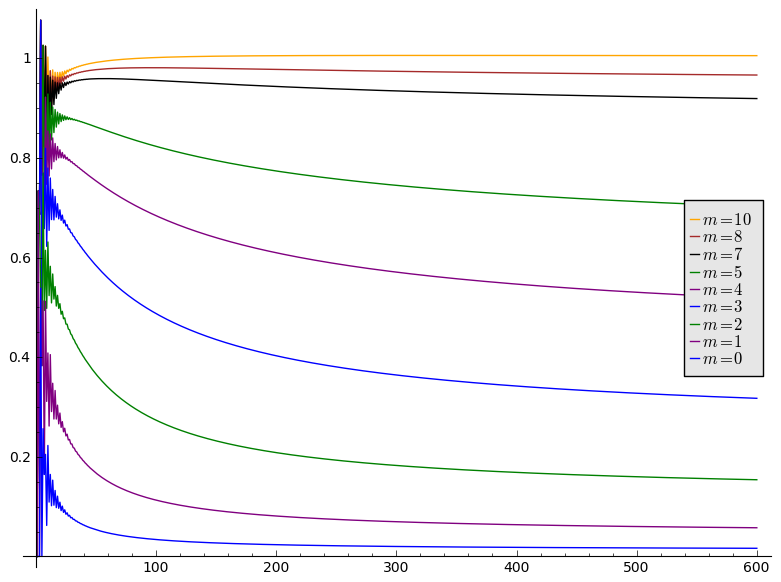}
  \caption{$S_{m,n} `r^n n^{3/2}$ up to $n=600$ for $m = 0$ to $10$}
  \label{fig:Smn_ran_n32}
\end{figure}

These numerical experiments allow us to state the following conjecture.

\begin{conj}
  For every natural number $m$, we have
  \[ S_{m,n} \sim 1.963448^n \cdot o \big( n^{-3/2} \big). \]
\end{conj}

\section{Unrankings}
\label{sec:unrank}

The recurrence relation \eqref{eq:Smn2} for $S_{m,n}$ allows us to define the
function generating $`l$-terms. More precisely, we construct bijections $s_{m,n}$,
called \emph{unranking} functions, between all non-negative integers not greater than
$S_{m,n}$ and binary $`l$-terms of size $n$ with at most $m$ distinct free
variables~\cite{al.:_rankin}. This approach is also known as the \emph{‘recursive
  method’}, originating with Nijenhuis and Wilf~\cite{nijenhuis78:_combin} (see
especially Chapter~13).  In order to describe unranking functions, we make use of the
Cantor pairing function.

Let us recall that for $n\ge2$ we have, by \eqref{eq:Smn2},
\begin{eqnarray*}
  S_{m,n} &=& S_{m+1,n-2} + \sum_{j=0}^{n-2} S_{m,j} S_{m,n-2-j} + [m \ge n-1].
\end{eqnarray*}
The encoding function $s_{m,n}$ takes an integer $k \in \{1,\ldots,S_{m,n}\}$ and
returns the term built in the following way.
\begin{itemize}
\item If $m \geq n-1$ and $k$ is equal to $S_{m,n}$, the function returns the string
  $1^{n-1} 0$.
\item If $k$ is less than or equal to $S_{m+1,n-2}$, then the corresponding term is
  in the form of abstraction $00 \widehat{\Var{M}}$, where $\widehat{\Var{M}}$ is the
  value of the unranking function $s_{m+1,n-2}$ on $k$.
\item Otherwise (i.e., $k$ is greater than $S_{m+1,n-2}$ and less than $S_{m,n}$ if
  $m\ge n+1$ or less than or equal to $S_{m,n}$ if $m < n+1$) then the corresponding
  term is in the form of application $01 \widehat{\Var{M}} \widehat{\Var{N}}$. In
  order to get strings $\widehat{\Var{M}}$ and $\widehat{\Var{N}}$, we compute the
  maximal value $\ell \in \{0,\ldots,n-2\}$ for which
$$k - S_{m+1,n-2} = \sum_{j=0}^{\ell-1} S_{m,j}S_{m,n-2-j} +r \qquad \textrm{with~} r \le
S_{m,\ell} S_{m,n-2-\ell}.$$%
The strings $\widehat{\Var{M}}$ and $\widehat{\Var{N}}$ are the values
$s_{m,\ell}(k')$ and $s_{m,n- 2-\ell}(k'')$, respectively, where $(k',k'')$ is the
pair of integers encoded by $r$ by the Cantor pairing function.
\end{itemize}
In Figure~\ref{fig:unrankT-prog} the reader may find a \textsf{Haskell}
program~\cite{jones03:_haskel} which computes the values $s_{m,n}(k)$.  In this
program, the function $s_{m,n}(k)$ is written as \texttt{unrankT m n k} and the
sequence $S_{m,n}$ is written as \texttt{tromp m n}.

\begin{figure}[!th]
  \begin{center}
    \hrule\medskip
    \begin{footnotesize}
\begin{verbatim}
unrankT :: Int -> Int -> Integer -> Term
unrankT m n k
  | m >= n - 1 && k == (tromp m n) = Index $ fromIntegral (n - 1) -- terms 1^{n-1}0
  | k <= (tromp  (m+1) (n-2)) = Abs (unrankT (m+1) (n-2) k) -- terms 00M
  | otherwise = unrankApp (n-2) 0 (k - tromp  (m+1) (n-2)) -- terms 01MN
    where unrankApp n j h
            | h <=  tmjtmnj  = let (dv,rm) = (h-1) `divMod` tmnj
                               in App (unrankT m j (dv+1)) (unrankT  m (n-j) (rm+1))
            | otherwise = unrankApp n (j + 1) (h -tmjtmnj) 
            where tmnj = tromp m (n-j)
                  tmjtmnj = (tromp m j) * tmnj 
\end{verbatim}
    \end{footnotesize}

    \medskip \hrule
  \end{center}
  \caption{A \textsf{Haskell} program for computing values of the function $s_{m,n}$}
  \label{fig:unrankT-prog}
\end{figure}

\section{Number of typable terms}\label{sec:typable}

The unranking function allows us to traverse all the closed terms of size $n$ and to
filter those that are typable (see~\cite{hindley97:_basic_simpl_theor} and appendix)
in order to count them and similarly to traverse all the terms of size $n$ to count
those that are typable.  
\ifAMS
Figure~\ref{fig:nb_typ} left gives the number $\nbct{n}$
of closed typable terms of size $n$ and Figure~\ref{fig:nb_typ} right gives the
number $\nbt{n}$ of all typable terms of size $n$. 
\begin{figure}[!th]
  \begin{scriptsize}
    \begin{center}
      \begin{math}
        \begin{array}{|l|l|}
          \hline
          \mathbf{n} & \nbct{n}\\\hline\hline
          0&0\\\hline
          1&0\\\hline
          2&0\\\hline
          3&0\\\hline
          4&1\\\hline
          5&0\\\hline
          6&1\\\hline
          7&1\\\hline
          8&1\\\hline
          9&1\\\hline
          10&5\\\hline
          11&4\\\hline
          12&9\\\hline
          13&13\\\hline
          14&23\\\hline
          15&29\\\hline
          16&67\\\hline
          17&94\\\hline
          18&179\\\hline
          19&285\\\hline
          20&503\\\hline
          21&795\\\hline
          22&1503\\\hline
          23&2469\\\hline
          24&4457\\\hline
          25&7624\\\hline
          26&13475\\\hline
          27&23027\\\hline
          28&41437\\\hline
          29&72165\\\hline
          30&128905\\\hline
          31&227510\\\hline
          32&405301\\\hline
          33&715078\\\hline
          34&1280127\\\hline
          35&2279393\\\hline
          36&4086591\\\hline
          37&7316698\\\hline
          38&13139958\\\hline
          39&23551957\\\hline
          40&42383667\\\hline
          41&76278547\\\hline
          42&137609116\\\hline
          43&248447221\\\hline
          44&449201368\\\hline
          45&812315229\\\hline
          46&1470997501\\\hline
        \end{array}
      \end{math}
      \qquad\qquad \mbox{
        \begin{math}
          \begin{array}{|l|l|}
            \hline
            \mathbf{n} & \nbt{n}\\\hline\hline
            0&0\\\hline
            1&0\\\hline
            2&1\\\hline
            3&1\\\hline
            4&2\\\hline
            5&2\\\hline
            6&3\\\hline
            7&5\\\hline
            8&8\\\hline
            9&13\\\hline
            10&22\\\hline
            11&36\\\hline
            12&58\\\hline
            13&103\\\hline
            14&177\\\hline
            15&307\\\hline
            16&535\\\hline
            17&949\\\hline
            18&1645\\\hline
            19&2936\\\hline
            20&5207\\\hline
            21&9330\\\hline
            22&16613\\\hline
            23&29921\\\hline
            24&53588\\\hline
            25&96808\\\hline
            26&174443\\\hline
            27&316267\\\hline
            28&572092\\\hline
            29&1040596\\\hline
            30&1888505\\\hline
            31&3441755\\\hline
            32&6268500\\\hline
            33&11449522\\\hline
            34&20902152\\\hline
            35&38256759\\\hline
            36&70004696\\\hline
            37&128336318\\\hline
            38&235302612\\\hline
            39&432050796\\\hline
            40&793513690\\\hline
            41& 1459062947\\\hline
            42& 2683714350\\\hline
          \end{array}
        \end{math}}
    \end{center}
  \end{scriptsize}

  \caption{Number of typable terms}
  \label{fig:nb_typ}
\end{figure}
\else
\cite{GrygielLescanne-Binary} provides tables of the numbers
of typable terms of size $n$. 
\fi

Thanks to the unranking function, we can build a \emph{uniform generator of
  $`l$-terms} and, using this generator, we can build a \emph{uniform generator of simply
  typable $`l$-terms}, which works by sieving the uniformly generated plain terms
through a program that checks their typability (see for
instance~\cite{grygiel_lescanne_jfp}).  This way, it is possible to generate
uniformly typable closed terms up to size $450$ which is rather good since Tromp was
able to build a self interpreter\footnote{Which is not typable by definition!} for
the $`l$-calculus of size $210$.
% \begin{figure}[!th]
%   \centering
%   \includegraphics[scale=.5]{numbers_of_typables_10_160.png}
%   \caption{Number of simply typable $`l$-terms and number of simply typable closed
%     $`l$-terms of size $n$ for $n$ from $10$ to $160$}
%   \label{fig:nb_typables}
% \end{figure}

\section{Conclusion}
\label{sec:concl}

We have shown that if we use the size yielded by the binary lambda
calculus~\cite{DBLP:conf/dagstuhl/Tromp06}, we get an exponential growth of the
number of $`l$-terms of size $n$ when $n$ goes to infinity.  This applies to closed
$`l$-terms, to $`l$-terms with a bounded number of free variables, and to all
$`l$-terms of size $n$. Except for the size of all $`l$-terms, the question of
finding the non-exponential factor of the asymptotic approximation of these numbers
is still open.  Since the generating functions are not standard, we were lead to
devise new methods for computing these approximations.  Beside, we describe unranking
functions (recursive methods) for generating $`l$-terms from which we derive tools
for their uniform generation and for the enumeration of typable $`l$-terms.  The
generation of random (typable) terms is limited by the performance of the generators
based on the recursive methods aka unranking which needs to handle huge number.
Boltzmann samplers~\cite{DBLP:journals/cpc/DuchonFLS04} should allow us to generate
terms of larger size.

\ifAMS
\appendix

\section{Types and typability}
\label{sec:typ}

Types determine whether $`l$-terms actually represent well-defined
functions~\cite{hindley97:_basic_simpl_theor}.  Here we focus on simple typable
terms, because simple typability is decidable.  Simple types are of two forms, either
variable types $`a$ or arrow types $`s"->"`t$:
\begin{displaymath}
  `s, `t \ ::= \ `a \mid `s "->" `t
\end{displaymath}
A context $`G = `t_n, \ldots, `t_1$ is a finite sequence of types which correspond to
declare that index $\Var{1}$ has type~$`t_1$, index $\Var{2}$ has type $`t_2$ etc. A
type judgment $`G "|-" M:`t$ says that in the context $`G$, the $`l$-term $M$ has
type $`t$.  To type a term we use inference rules:
\begin{displaymath}
  \prooftree
  \justifies  `t_n, \ldots, `t_i, \ldots `t_1 "|-" \Var{i}: `t_i
  \using Var
  \endprooftree
  \qquad\qquad
  \prooftree
  `G, `t "|-" M: `s 
  \justifies `G "|-" `l M : `t "->" `s
  \using Abs
  \endprooftree
  \end{displaymath}

  \begin{displaymath}
    \prooftree
    `G "|-" M: `s "->" `t \qquad `G"|-" P : `s
    \justifies `G "|-" MP: `t
    \using App
    \endprooftree
  \end{displaymath}
\begin{definition}[Typability]
  A term $M$ is \emph{typable} if there exists a context $`G$ and type $`s$ such that
  $`G "|-" M: `s$.
\end{definition}
Notice that an open term with $n$ free indices require a context of size $n$ to be
typable. Therefore a closed term requires an empty context to be typable. Moreover
checking typability is solving constraints, mostly constraints generated by rule
\emph{App}.  For instance, term $`l \Var{1} \Var{1}$ cannot be typed since $\Var{1}$
of type say $`s$ cannot be applied to the term $\Var{1}$ of type $`s$. For that it
should be of type $`s "->" `t$.  Similarly $`l `l \Var{2} \Var{1}\Var{1}$ of type
$(`a "->" `a "->" `b) "->" `a "->" `b$ cannot be applied to $`l\Var{1}$ of type $`g
"->"`g$.  Therefore $(`l `l \Var{2} \Var{1}\Var{1})\,`l\Var{1}$ is not typable.  We
also notice that typability can be described neither recursively nor structurally.

\bigskip
\fi
\end{document}

From the program for checking  typability of terms and from the unranking program we where able to compute the ratio of simply typable terms over terms as shown in Figure~\ref{fig:ratio}. The  computation was made on a poor laptop.  Until $35$ the computation was made  exhaustively because $S_{0,35} = 5\,977\,863$ and $S_{0,36} = 11\,148\,652$ and after using a Monte Carlo method.  
 \begin{figure}
   \centering
   \begin{small}
     \begin{\inftyath}
       \begin{array}{|c||c|c|c|c|c|c|c|c|c|c|c|c|c|c|c|c|c|c|c|c|c|c|c|c|c|c|c|c|c|}
         \hline
         n &1 & 2 & 3 & 4 & 5 & 6 & 7 & 8 & 9 & 10 & 11 &12 & 13 & 14 & 15 & 16 & 17 & 18
         &19 \\\hline\hline
         \mathfrak{T}(n) & 0 & 0 & 0 & 1 & 0 & 1 & 1 & 1 & 1 & 5 & 4 & 9 & 13 & 23 & 29 & 67 & 94 & 179
         & 285
         \\\hline
       \end{array}
       \end{\inftyath}
       \medskip

       \begin{\inftyath}
         \begin{array}{|c||c|c|c|c|c|c|c|c|c|c|c|c|c|c|c|c|c|c|}
           \hline
           n & 20 & 21 & 22 & 23 & 24 & 25 & 26 & 27 & 28 &29 \\\hline\hline
           \mathfrak{T}(n) & 503 & 795 & 1503 & 2469 & 4457 & 7624 & 13475 & 23027 &
           41437 & 72165 
           \\\hline
         \end{array}
       \end{\inftyath}
 
       \medskip

       \begin{\inftyath}
         \begin{array}{|c||c|c|c|c|c|c|c|c|c|c|c|c|c|c|c|c|c|c|}
           \hline
           n & 30 & 31 & 32 & 33 & 34 & 35 \\\hline\hline
           \mathfrak{T}(n)& 128905 & 227510 & 405301 & 715078 &1280127 & 2279393
           \\\hline
         \end{array}
       \end{\inftyath}
   \end{small}
   \caption{The numbers $\mathfrak{T}(n)$ of typable terms of size $n$}
   \label{fig:nb_typ}
 \end{figure}

\begin{figure}
  \centering
   \includegraphics[width=\columnwidth]{ratio_typables.png}
  \caption{Ratio of typable terms over all terms}
  \label{fig:ratio}
\end{figure}

\section{Weight $1$ for abstractions and applications}

Assume that now we take the following size for terms.
 \begin{eqnarray*}
    |\Var{n}| &=& n +1\\
    |`l M | &=& |M| + 1\\
    |M\,N| &=& |M| + |N|+1
  \end{eqnarray*}
Then
\begin{eqnarray*}
  T_{\infty,0} &=& 0\\
  T_{\infty,n+1} &=& [m\ge n+1] + T_{\infty+1,n} + \sum_{k=0}^n T_{\infty,k} T_{\infty,n-k}
\end{eqnarray*}
and
\begin{eqnarray*}
  T_{\infty,n+1} &=& 1 + T_{\infty,n+1} + \sum_{k=0}^n  T_{\infty,k} T_{\infty,n-k}
\end{eqnarray*}
So
\begin{eqnarray*}
  \Tt_\infty(z) &=& \frac{z}{1-z} + z\Tt_\infty(z) + z\Tt_\infty^2(z)
\end{eqnarray*}
hence 
\begin{displaymath}
  z\Tt_\infty^2(z) - (1-z) \Tt_\infty(z) +  \frac{z}{1-z} = 0
\end{displaymath}
and 
\begin{displaymath}
  `D = (1-z)^2 + \frac{4z^2}{1-z}
\end{displaymath}
so the singularity we look for is a root of polynomial
\begin{displaymath}
  (1-z)^3 - 4z^2
\end{displaymath}
which is about $0.2955977425221616$ with inverse $3.3829757679053585$.
We have
\begin{displaymath}
  \Tt_m = \frac{1 - \sqrt{1-4z^2(\frac{1-z^m}{1-z} + \Tt_{\infty+1})}}{2z}
\end{displaymath}

\begin{thebibliography}{10}

\bibitem{bodini11:_lambd_bound_unary_heigh}
Olivier Bodini, Dani{\`e}le Gardy, and Bernhard Gittenberger.
\newblock Lambda-terms of bounded unary height.
\newblock {\em 2011 Proceedings of the Eighth Workshop on Analytic Algorithmics
  and Combinatorics (ANALCO)}, 2011.

\bibitem{2013arXiv1305.0640B}
Olivier {Bodini}, Dani{\`e}le {Gardy}, Bernhard {Gittenberger}, and Alice
  {Jacquot}.
\newblock {Enumeration of generalized $BCI$ lambda-terms}.
\newblock {\em ArXiv e-prints}, May 2013.

\bibitem{DBLP:conf/icfp/ClaessenH00}
Koen Claessen and John Hughes.
\newblock \textsf{QuickCheck}: a lightweight tool for random testing of
  \textsf{Haskell} programs.
\newblock In Martin Odersky and Philip Wadler, editors, {\em ICFP}, pages
  268--279. ACM, 2000.

\bibitem{DBLP:journals/corr/abs-0903-5505}
Ren{\'e} David, Katarzyna Grygiel, Jakub Kozik, Christophe Raffalli, Guillaume
  Theyssier, and Marek Zaionc.
\newblock Asymptotically almost all $\lambda$-terms are strongly normalizing.
\newblock {\em Logical Methods in Computer Science}, 9(1:02):1--30, 2013.

\bibitem{deBruijn72}
Nicolaas~G. de~Bruijn.
\newblock Lambda calculus notation with nameless dummies, a tool for automatic
  formula manipulation, with application to the {Church-Rosser} theorem.
\newblock {\em Indagationes Mathematicae}, 34(5):381--392, 1972.

\bibitem{DBLP:journals/cpc/DuchonFLS04}
Philippe Duchon, Philippe Flajolet, Guy Louchard, and Gilles Schaeffer.
\newblock Boltzmann samplers for the random generation of combinatorial
  structures.
\newblock {\em Combinatorics, Probability {\&} Computing}, 13(4-5):577--625,
  2004.

\bibitem{al.:_rankin}
A.~Karttunen et~al.
\newblock Ranking and unranking functions.
\newblock OEIS Wiki.
\newblock \url{http://oeis.org/wiki/Ranking_and_unranking_function}.

\bibitem{flajolet08:_analy_combin}
Philippe Flajolet and Robert Sedgewick.
\newblock {\em Analytic Combinatorics}.
\newblock Cambridge University Press, 2008.

\bibitem{grygiel_lescanne_jfp}
Katarzyna Grygiel and Pierre Lescanne.
\newblock Counting and generating lambda terms.
\newblock {\em Journal of Functional Programming}, to appear, 2013.

\bibitem{GrygielLescanne-Binary}
Katarzyna Grygiel and Pierre Lescanne.
\newblock Counting terms in the binary lambda calculus.
\newblock Technical report, Arxiv, 2013.

\bibitem{hindley97:_basic_simpl_theor}
J.~Roger Hindley.
\newblock {\em Basic Simple Type Theory}.
\newblock Number~42 in Cambridge Tracts in Theoretical Computer Science.
  Cambridge University Press, 1997.

\bibitem{jones03:_haskel}
Simon~Peyton Jones, editor.
\newblock {\em Haskell 98 language and libraries: the Revised Report}.
\newblock Cambridge University Press, 2003.

\bibitem{LescannePOPL94}
Pierre Lescanne.
\newblock From $\lambda\sigma$ to $\lambda\upsilon$, a journey through calculi
  of explicit substitutions.
\newblock In Hans Boehm, editor, {\em Proceedings of the 21st Annual {ACM}
  Symposium on {P}rinciples {O}f {P}rogramming {L}anguages, {Portland (Or.,
  USA)}}, pages 60--69. ACM, 1994.

\bibitem{DBLP:journals/tcs/Lescanne13}
Pierre Lescanne.
\newblock On counting untyped lambda terms.
\newblock {\em Theor. Comput. Sci.}, 474:80--97, 2013.

\bibitem{LiVitanyi}
Ming Li and Paul Vit\'{a}nyi.
\newblock {\em An introduction to Kolmogorov complexity and its applications
  (3rd ed.)}.
\newblock Springer-Verlag New York, Inc., 2008.

\bibitem{Mitchell96}
John~C. Mitchell.
\newblock {\em Foundations for Programming Languages}.
\newblock {MIT} Press, sep 1996.

\bibitem{nijenhuis78:_combin}
Albert Nijenhuis and Herbert~S. Wilf.
\newblock {\em Combinatorial algorithms, 2nd edition}.
\newblock Computer science and applied mathematics. Academic Press, New York,
  1978.

\bibitem{DBLP:conf/dagstuhl/Tromp06}
John Tromp.
\newblock Binary lambda calculus and combinatory logic.
\newblock In Marcus Hutter, Wolfgang Merkle, and Paul M.~B. Vit{\'a}nyi,
  editors, {\em Kolmogorov Complexity and Applications}, volume 06051 of {\em
  Dagstuhl Seminar Proceedings}. Internationales Begegnungs- und
  Forschungszentrum fuer Informatik (IBFI), Schloss Dagstuhl, Germany, 2006.

\bibitem{wang04:_effic_gener_random_progr_their_applic}
Jue Wang.
\newblock The efficient generation of random programs and their applications.
\newblock Honors Thesis, Wellesley College, Wellesley, MA, May 2004.

\end{thebibliography}
\end{document}